\documentclass[letterpaper, 10pt, conference]{LatexTemplate/ieeeconf}  

\IEEEoverridecommandlockouts                              
\usepackage{cite}
\usepackage[pdftex]{graphicx}

\usepackage{amsmath}
\usepackage{amssymb}
  \usepackage{amsthm}   
\usepackage{amsfonts}       
\usepackage[section]{placeins}
\usepackage{algorithm}
\usepackage[noend]{algpseudocode}
\usepackage{url}            
\usepackage[utf8]{inputenc} 
\usepackage[T1]{fontenc}    
\usepackage{booktabs}       
\usepackage{nicefrac}       
\usepackage{microtype}      

\makeatletter
\let\NAT@parse\undefined
\makeatother

\usepackage[colorlinks=false, pdfborder={0 0 0.5}]{hyperref}
\usepackage{cleveref}		

\newtheorem{thm}{Theorem}
\newtheorem{lem}{Lemma}
\newtheorem{rem}{Remark}
\newtheorem{defn}{Definition}
\newtheorem{cor}{Corollary}\crefname{cor}{Corollary}{Corollaries}
\newtheorem{prb}{Problem}\crefname{prb}{Problem}{Problems}

\makeatletter
\algrenewcommand\ALG@beginalgorithmic{\footnotesize}
\makeatother

\title{\LARGE \bf
Shared Linear Quadratic Regulation Control:\\A Reinforcement Learning Approach
}

\author{Murad~Abu-Khalaf, Sertac~Karaman, and~Daniela~Rus
	\thanks{Murad~Abu-Khalaf and Daniela~Rus are with the MIT Computer Science and Artificial Intelligence Laboratory, Massachusetts Institute of Technology, Cambridge, MA~02139 USA. {\tt\small \{murad,rus\}@csail.mit.edu}}%
	\thanks{Sertac~Karaman is with the Laboratory of Information and Decision Systems, Massachusetts Institute of Technology, Cambridge, MA~02139 USA. {\tt\small sertac@mit.edu}}%
	\thanks{We gratefully acknowledge the support provided by the Toyota Research Institute.}%
}

\begin{document}
\allowdisplaybreaks

\maketitle
\thispagestyle{empty}
\pagestyle{empty}

\begin{abstract}

We propose controller synthesis for state regulation problems in which a human operator shares control with an autonomy system, running in parallel. The autonomy system continuously improves over human action, with minimal intervention, and can take over full-control if necessary. It additively combines user input with an adaptive optimal corrective signal to drive the plant. It is adaptive in the sense that it neither estimates nor requires a model of the human's action policy, or the internal dynamics of the plant, and can adjust to changes in both. Our contribution is twofold; first, a new controller synthesis for shared control which we formulate as an adaptive optimal control problem for continuous-time linear systems and solve it online as a human-in-the-loop reinforcement learning. The result is an architecture that we call \emph{shared} linear quadratic regulator (sLQR). Second, we provide new analysis of reinforcement learning for continuous-time linear systems in two parts. In the first analysis part, we avoid learning along a single state-space trajectory which we show leads to data collinearity under certain conditions. In doing so, we make a clear separation between exploitation of learned policies and exploration of the state-space, and propose an exploration scheme that requires switching to new state-space trajectories rather than injecting noise continuously while learning. This avoidance of continuous noise injection minimizes interference with human action, and avoids bias in the convergence to the stabilizing solution of the underlying algebraic Riccati equation. We show that exploring a minimum number of pairwise distinct state-space trajectories is necessary to avoid collinearity in the learning data. In the second analysis part, we show conditions under which existence and uniqueness of solutions can be established for off-policy reinforcement learning in continuous-time linear systems; namely, prior knowledge of the input matrix.

\end{abstract}

\section{Introduction} \label{Section:Intro}

We address technical challenges associated with adaptation and optimization that emerge when a human operator shares low-level regulation control tasks with an automatic control system. The rapid acceleration in machine learning and artificial intelligence at large creates opportunities to operate in a data-driven response-based manner, rather than an offline model-based design, and without a priori knowledge of the human-in-the-loop action policy. We wish to investigate these notions for a fundamental building block of control system design, \textit{e.g.} the linear quadratic regulator (LQR). The theory is inspired by applications in which a human operator, assisted by an autonomy system, is regulating a steering angle, speed, or spacing between vehicles, among others, or balancing objects including the human body itself.

Different levels of shared autonomy are reported in the literature. In robotic teleoperation, shared autonomy augments a user's ability to control a robot via an interface that generally has less degrees of freedom than the robot itself. Many such paradigms of shared autonomy require the autonomy system to predict user intentions and augment user input via a policy arbitration scheme. The authors in \cite{Dragan_PolicyBlending2013} proposed policy-blending as a ``common lens'' to understand policy arbitration and control effort division across time or tasks while \cite{SharedAutonomyHindsightOptimization2015} proposes to handle a distribution of goals at once by formulating a Partially Observable Markov Decision Problem (POMDP). Modeling assumptions for the underlying POMDP and the user policy have been relaxed in \cite{Dragan_SharedAutonomyDeepReinforcementLearning2018} by leveraging a human-in-the-loop deep Q-learning.

In \cite{Anderson_OptimalControlSemi-autonomous2010, Anderson_ConstraintSemi-autonomous2012}, shared control for semi-autonomous driving is formulated as a Model Predictive Control (MPC) problem. Its cost function does not account for human input and uses instead a form of policy arbitration. In \cite{Tilbury_SharedControl2017}, an MPC problem to deal with obstacle avoidance is formulated that handles policy arbitration implicitly, and uses quadratic costs that consider human input. In \cite{Schwarting_MPC2017}, parallel autonomy for safe driving is formulated as an MPC problem that handles policy arbitration implicitly as well. It estimates the user input and holds it constant while evolving the dynamics over the prediction horizon. It accounts for the user input in the cost function and uses a forgetting factor to emphasize short-term impact of user input. All these MPC approaches are model-based, and human input when considered is provided as an estimate that is held constant over the prediction horizon.

In \cite{Annaswamy_SharedAdaptiveAutoPilot2018}, aircraft control tasks are shared between a human pilot and an adaptive autopilot. The pilot assumes high-level tasks, detection of anomalies and switching of controller structure, and relegates low-level regulation to the autopilot.

In this paper, a human operator shares low-level regulation control task with an adaptive parallel autonomy system via a human-computer interface (HCI), and without requiring a model of the human policy or the plant's internal dynamics. We formulate the problem in an optimal control theoretic sense, and solve the underlying dynamic programming problem online via reinforcement learning. The objective of the autonomy system is to assist the human operator by improving closed-loop performance without significantly deviating from user input, and to take over from the user when necessary. It serves as a secondary controller or co-pilot, that optimally modulates with an additive corrective signal the output of a primary controller enacted by a human operator; therefore, policy arbitration here is implicit. We refer to the resulting architecture as a shared linear quadratic regulator (sLQR). It solves the underlying algebraic Riccati equation (ARE) for the human-in-the-loop closed-loop dynamics.

Our learning approach is driven by the system response and requires partial model information. We dissect the role of exploration from exploitation to eliminate sources of collinearity in the learning data, and thus the need to continuously inject perturbation noise as common in the literature. Our focus is on continuous-time linear systems. A quadratic form is used to capture the Lyapunov function or cost-to-go of the optimization problem resulting in essentially learning algorithms that involve polynomial regression. For more general nonlinear dynamics with state and input constraints, extensions are possible via the use of neural networks and approximate dynamic programming tools from \cite{AbukhalafLewis_NNHJB2005}.

In \Cref{Section:Notation}, we introduce the notation used throughout the paper. In \Cref{Section:ProblemFormulation}, we formulate shared control as an adaptive optimal control problem for state regulation. In \Cref{Section:PreliminariesPolicyIterations} we review relevant background on reinforcement learning useful to our shared control methodology and point out existing gaps in the literature. In \Cref{Section:NewSolvabilityAnalysis} we address some existing gaps by providing new solvability analysis for policy iterations. In \Cref{Section:sLQR}, we solve the formulated shared control problem and present the sLQR. In \Cref{Section:Example} we apply sLQR to a car-following application. In Section \ref{Section:Conclusion}, we provide some conclusions and future directions.

\subsection{Notation} \label{Section:Notation}
$\mathbb{R}$ denotes the real line and $\mathbb{C}$ the set of complex numbers. Given matrix $Y\in \mathbb{R}^{p\times q} $, $Y_{([1:p],j)} $ denotes its \textit{j\textsuperscript{th}} column and $Y_{([i_1:i_2],j)} $ denotes a column vector formed from the \textit{j\textsuperscript{th}} column of $Y$ starting at row $i_1$ and ending at row $i_2$. The element at the \textit{i\textsuperscript{th}} row and \textit{j\textsuperscript{th}} column is denoted by $Y_{(i,j)} $. From \cite{Brewer1978}, $vec(\cdot )$ stacks the columns of $Y$ from first to last into a single column of size $pq\times 1$ as shown, and the operator $\otimes $ is the Kronecker product where given matrix $Z$ of an arbitrary size
\[vec(Y)\buildrel\Delta\over= \left[\begin{array}{c} {Y_{([1:p],1)} } \\ {Y_{([1:p],2)} } \\ {\vdots } \\ {Y_{([1:p],q)} } \end{array}\right], 
Y\otimes Z\buildrel\Delta\over= \left[\begin{array}{ccc} {Y_{(1,1)} Z} & {\ldots } & {Y_{(1,q)} Z}  \\ {\vdots }  & {\ddots } & {\vdots } \\ {Y_{(p,1)} Z}  & {\ldots } & {Y_{(p,q)} Z} \end{array}\right].\]

We define $vec^{L} (\cdot )$ which operates on a square matrix of size $n\times n$ by stacking from first to last the columns of its lower triangular part into a single column of size $\frac{n(n+1)}{2}\times 1$ as shown for $Z\in \mathbb{R}^{n\times n}$
\[vec^{L}(Z)\buildrel\Delta\over= \left[\begin{array}{c} {Z_{([1:n],1)} } \\ {Z_{([2:n],2)} } \\ {\vdots } \\ {Z_{([n:n],n)} } \end{array}\right].\]

\noindent $\kappa(i,j,n)$, where $i\geq j$, denotes the index of an element at the \textit{i\textsuperscript{th}} row and \textit{j\textsuperscript{th}} column of $Z\in \mathbb{R}^{n\times n}$ as it appears in $vec^{L} (Z )$. For example, $\kappa(1,1,n)=1$ and $\kappa(n,n,n)=\frac{n(n+1)}{2}$.

If $x \in \mathbb{R}^{n\times 1}$, then $x\otimes x = vec(x{x}^{\intercal})$ where ${x}^{\intercal}$ is the transpose of $x$. If additionally $P\in \mathbb{R}^{n\times n}$, then $x^{\intercal}Px =(x\otimes x)^{\intercal} vec(P)=vec(xx^{\intercal} )^{\intercal} vec(P)=vec^{L} (xx^{\intercal} )^{\intercal} W$, where $W\in \mathbb{R}^{\frac{n(n+1)}{2}\times 1}$ and $W_{(\kappa(i,j,n),1)} =  P_{(i,j)}+P_{(j,i)}-\delta_{ij}P_{(i,j)}$ where $\delta_{ij}$ is the Kronecker delta function. This follows from matching the terms of
\begin{equation*}
\begin{aligned}
vec^{L} (xx^{\intercal} )^{\intercal} W&=\sum _{j=1}^{n}\sum _{i=j}^{n}w_{\kappa(i,j,n)} x_{i} x_{j},\\
vec(xx^{\intercal} )^{\intercal} vec(P)&=\sum _{j=1}^{n}\sum _{i=j}^{n}\left( P_{(i,j)}+P_{(j,i)}-\delta_{ij}P_{(i,j)}\right)  x_{i} x_{j}.
\end{aligned}	
\end{equation*}

An all zero-entries $n\times m$ matrix is denoted by $\mathbf{0}_{n\times m} $, and by $\mathbf{0}$ when the size is context-dependent. Similarly, $\mathbf{I}_{n} $ denotes a size \textit{n} identity matrix, and $\mathbf{I}$ is used when the size is context-dependent. Moreover, we denote by $\mathbf{e}_{ij} \in \mathbb{R}^{m \times n}$ a matrix with $1$ at the $i^{th}$ row and $j^{th}$ column and $0$ otherwise.

\section{Problem Formulation} \label{Section:ProblemFormulation}

Consider a continuous-time linear system given by
\begin{subequations}  \label{Equation:HCI}
	\begin{align}\label{Equation:LinearSystem}
	\dot{x}(t)=&Ax(t)+Bu(t),\\
	y_a =& x,\\
	y_h =& C_h x,\\
	u=&u_{h} + u_{a}, \label{Equation:DecomposedInput}
	\end{align}
\end{subequations} 
with $x\in {\mathbb{R}}^{n\times 1} $ and $u\in {\mathbb{R}}^{m\times 1} $; $A\in {\mathbb{R}}^{n\times n} $, $B\in {\mathbb{R}}^{n\times m}$ and $C_h \in {\mathbb{R}}^{p\times n}$ are constant matrices with $(A,B)$ stabilizable, and $(A,B,C_h)$ static output feedback stabilizable.  $y_a \in {\mathbb{R}}^{n\times 1}$ is the plant output available to the autonomy system, while $y_h \in {\mathbb{R}}^{p\times 1} $ is the plant output available to the human operator. The input to the plant $u$ is provided by a human-computer interface and is decomposed into a human-generated $u_{h}$ and an autonomy-computed $u_{a}$ per \eqref{Equation:DecomposedInput}.

\begin{figure}[!htb]
	\center{\includegraphics[width=\columnwidth]
		{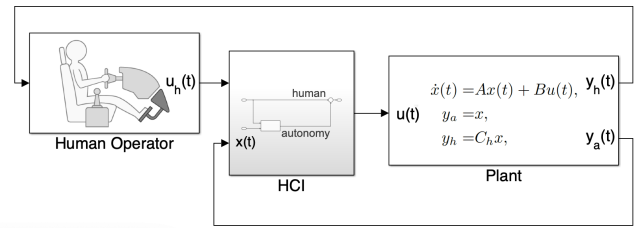}}
	\caption{\label{Figure:HCI} Parallel Autonomy.}
\end{figure}
The human's action policy is unknown to the autonomy system which can observe and measure $u_{h}$. It is assumed that the human operator observes the plant via $y_h = C_h x$ and that the human implements a linear static output feedback control to stabilize the plant as follows $u_{h} (x)=K_{h}y_h(x) =  K_{h} C_hx$. Thus, the human policy is linear in the state for an observation matrix $C_h$ and gain $K_{h}$ that are unknown to the autonomy system. Unlike the human operator, the autonomy system has full access to the state of the plant, \textit{i.e.} $y_a = x$.

The aim of the autonomy system is to compute and apply $u_{a} $ to meet two objectives: First, to improve the human-in-the-loop closed-loop performance with minimum intervention as experienced by the user. Second, to entirely take over control from the human operator when necessary. These two objectives are stated concisely as Problem 1 and Problem 2.

\begin{prb} \label{Problem:MainProblem}
	\emph{Minimum Intervention}: Consider the parallel autonomy system \eqref{Equation:HCI}, and assume no prior knowledge of the plant's internal dynamics matrix $A$, the human's observation matrix $C_h$ and the human output feedback gain $K_{h}$. Solve the infinite-horizon optimal control problem
	
	\begin{equation} \label{Equation:MainProblem} 
	J(x_{0}, t_0, u_{h})={\inf\limits_{u_a}} \smallint\limits_{t_{0}}^{\infty}\left(x^{\intercal} Qx+u_{h}^{\intercal} (x)Mu_{h} (x)+u_{a}^{\intercal} Ru_{a} \right)dt,
	\end{equation} 
	where $x_0=x(0)$, adapting to subsequent changes in $A$,  $K_{h}$, and $C_h$. It is assumed that $u_h$ is stabilizing.
\end{prb}

\begin{prb} \label{Problem:TakeOver}
	\emph{Takeover}: Consider the parallel autonomy system \eqref{Equation:HCI}, and assume no prior knowledge of $A$, $C_h$ and $K_{h}$. With a human-in-the-loop, $K_{h}\neq0$ and stabilizing, learn an optimal control policy for the infinite-horizon problem \eqref{Equation:MainProblem} such that the learned policy is optimal for $u_{h}=0$, \emph{i.e.} is able to stabilize the plant and take over control entirely after the human exits the loop.
\end{prb}

\section{ Policy Iterations:  Gaps in Previous Results} \label{Section:PreliminariesPolicyIterations}

In Section \ref{Section:ResponseBased}, \Cref{Remark:Collinearity} and \Cref{Remark:Nonunique} highlight gaps in response-based policy iteration schemes useful to our proposed shared control approach. We address these gaps in Section \ref{Section:NewSolvabilityAnalysis}. We first briefly revisit offline policy iterations.

\subsection{Offline Model-Based Policy Iterations}  \label{Section:Offline}

This type of policy iterations is executed offline without generating a system response. It requires full-knowledge of system matrices $A$ and $B$. In this case, given $u_{i}(t)=K_{i} x(t)$ with stabilizing $K_{i} \in {\mathbb{R}}^{m \times n} $, one has state-space trajectories $\varphi(t,x_0,u_{i}(t))$ where $x_0 = x(0)$ and such that
\begin{equation} \label{Equation:StateTrajectories} 
x(t)=e^{(A+BK_{i} )(t-t_{0} )} x(t_{0}).
\end{equation} 
A quadratic cost-to-go for policy $u_{i}$ is obtained by integrating over $\varphi(t,x_0,u_{i}(t))$ as follows
\begin{equation} \label{Equation:Cost2Go} 
\begin{aligned}
J_i(x_{0}, t_0) &=\smallint\limits_{{{t}_{0}}}^{\infty }{\left( {{x}^{\intercal}}(t)Qx(t)+u_{i}^{\intercal}(t)R{{u}_{i}}(t) \right)dt} \\ 
& = x{{({{t}_{0}})}^{\intercal}}{{P}_{i}}x({{t}_{0}}) = {{V}_{i}}(x({{t}_{0}})) ,
\end{aligned}
\end{equation} 
where
\begin{equation} \label{Equation:Cost2GoMatrix} 
P_{i} = \smallint\limits_{t_{0} }^{\infty } \left[e^{(A+BK_{i} )(t-t_{0} )} \right]^{\intercal} \left[Q+K_{i}^{\intercal} RK_{i} \right]e^{(A+BK_{i} )(t-t_{0} )}dt.
\end{equation} 

\noindent It follows from \eqref{Equation:Cost2GoMatrix} that $P_{i}$ satisfies the Lyapunov equation
\begin{multline} \label{Equation:PiLyapunov}
{P_{i} (A+BK_{i} )+(A+BK_{i} )^{\intercal} P_{i} } =\\ 
\smallint\limits_{t_{0} }^{\infty }\textstyle\frac{d}{dt} \left(\left[e^{(A+BK_{i} )(t-t_{0} )} \right]^{\intercal} \left[Q+K_{i}^{\intercal} RK_{i} \right]e^{(A+BK_{i} )(t-t_{0} )} \right)dt  \\
=-\left[Q+K_{i}^{\intercal} RK_{i} \right].
\end{multline} 
$V_{i} (x)$ is a Lyapunov function for $u_{i}(x)$ satisfying \eqref{Equation:ViLyapunov}
\begin{subequations} \label{Equation:PolicyIteration}
	\begin{align}
	{\textstyle\frac{dV_{i} }{dx}} ^{\intercal} \left(Ax+Bu_{i} \right)&=-x^{\intercal} Qx-u_{i}^{\intercal} Ru_{i}, \label{Equation:ViLyapunov} \\
	u_{i+1} &=-{\textstyle\frac{1}{2}} R^{-1} B^{\intercal} {\textstyle\frac{dV_{i} }{dx}}, \label{Equation:PolicyUpdate}
	\end{align} 
\end{subequations}
and \eqref{Equation:PolicyUpdate} is the policy iteration update.

It was shown in \cite{Kleinman1968} that iterating on \eqref{Equation:KleinmanIteration} and \eqref{Equation:KleinmanUpdate}, with proper initialization, converges to the stabilizing solution of the algebraic Riccati equation \eqref{Equation:ARE} denoted by $\mathcal{P}$
\begin{subequations} \label{Equation:KleinmanRiccati}
	\begin{gather}
	P_{i} (A+BK_{i} )+(A+BK_{i} )^{\intercal} P_{i} +K_{i}^{\intercal} RK_{i} +Q=0, \label{Equation:KleinmanIteration} \\
	K_{i+1} =-R^{-1} B^{\intercal} P_{i}, \label{Equation:KleinmanUpdate}  \\
	PA+A^{\intercal} P-P^{\intercal} B^{\intercal} R^{-1} B^{\intercal} P+Q=0,  \label{Equation:ARE}
	\end{gather} 
\end{subequations}
where $\mathcal{P}$ is the quadratic value function  matrix
\begin{equation*}
{x({t}_{0})}^{\intercal}\mathcal{P}x(t_0) = \inf\limits_{u} {\smallint\limits_{t_{0} }^{\infty }\left(x^{\intercal} (t)Qx(t)+u^{\intercal} Ru\right)dt  }.
\end{equation*}

\subsection{Response-Based Data-Driven Policy Iterations}  \label{Section:ResponseBased}
Let $u_i$ be the \emph{target policy}, \emph{i.e.} the one we would like to improve by policy iterations and learn its cost-to-go. Let $u_b$ be a \emph{behavior policy}, \emph{i.e.} one generating the system response. In on-policy learning, $u_b = u_i$ while in off-policy learning $u_b \neq u_i$. Off-policy learning enables experience replays whereby available trajectories generated in the past by other policies can be used to improve a target policy.

\subsubsection{On-Policy}
In \cite{Vrabie_2009}, the right-hand-side of \eqref{Equation:Cost2Go} is split into two parts
\begin{equation} \label{Equation:IntegralReward}
\begin{aligned}
{{V}_{i}}(x({{t}_{k}}))  =&\smallint\limits_{{{t}_{k}}}^{{{t}_{k}}+\tau }{\left( {{x}^{\intercal}}(t)Qx(t)+u_{i}^{\intercal}(t)R{{u}_{i}}(t) \right)dt} \\
& +{{V}_{i}}(x({{t}_{k}}+\tau ) ).
\end{aligned}
\end{equation}
Substitute $V_i(x)=x^{\intercal}P_ix$ in \eqref{Equation:IntegralReward} to get
\begin{multline} \label{Equation:PiIntegralReward} 
x{{({{t}_{k}}+\tau )}^{\intercal}}{{P}_{i}}x({{t}_{k}}+\tau )-x{{({{t}_{k}})}^{\intercal}}{{P}_{i}}x({{t}_{k}}) \\ 
=-\smallint\limits_{{{t}_{k}}}^{{{t}_{k}}+\tau }{\left( {{x}^{\intercal}}(t)Qx(t)+u_{i}^{\intercal}(t)R{{u}_{i}}(t) \right)dt}.
\end{multline}
The result is a response-based data-driven approach to solve for $P_i$ given online measurements. This online data, indexed by $k$, has two parts; the first part is state measurements $x(t_{k} )$ and $x(t_{k}+\tau )$ on the left-hand side of \eqref{Equation:PiIntegralReward}; the second is the reinforcement signal experienced over a finite horizon $[t_{k} ,t_{k} +\tau ]$ as represented by the integral on the right-hand side of \eqref{Equation:PiIntegralReward}. This results in a linear system of unknowns, namely the coefficients of a quadratic form, that can be solved for by polynomial regression. Unlike \eqref{Equation:KleinmanIteration}, equation \eqref{Equation:PiIntegralReward} does not require $A$ and $B$ of \eqref{Equation:LinearSystem} to solve for $P_i$. The policy update \eqref{Equation:KleinmanUpdate} still requires the knowledge of $B$.

In \cite{Vrabie_2009}, the linear system of unknowns \eqref{Equation:PiIntegralReward} is written element-wise using Kronecker products such that the unknowns, \textit{i.e.} the elements of $P_i$, form a vector. For the purpose of polynomial regression, dependent regressors should be eliminated, \textit{i.e.} $x_ix_j$ and $x_jx_i$ are dependent. The lower triangular part of $xx^{\intercal} $ represents all independent regressors. In this regard, and using the notation of Section \ref{Section:Notation}, we use a quadratic form $V(x) = {\mathbf{w}}^{\intercal}\phi(x)$, where $\phi(x)=vec^{L} (xx^{\intercal} )$, to solve for $\mathbf{w}$ from measurements by forming the following linear system of unknowns $\mathbf{A}\mathbf{w} = \mathbf{b}$. The measurement matrices $\mathbf{A} \in \mathbb{R}^{N\times \frac{n(n+1)}{2}}$ and $\mathbf{b} \in \mathbb{R}^{N\times 1}$ are
\begin{subequations} \label{Equation:Ab_datum}
	\begin{align}
	\mathbf{A} &=  \left[\begin{array}{l} \label{Equation:A_datum}
	\phi(x_{\tau}[0])^{\intercal} - \phi(x[0])^{\intercal}\\ 
	\vdots\\
	\phi(x_{\tau}[k])^{\intercal} - \phi(x[k])^{\intercal}\\
	\vdots\\
	\phi(x_{\tau}[N-1])^{\intercal} - \phi(x[N-1])^{\intercal}
	\end{array}\right],\\
	\mathbf{b} &=  -\left[\begin{array}{l}  \label{Equation:b_datum}
	r_0, \hdots, r_k, \hdots, r_{N-1}
	\end{array}\right]^{\intercal},\\
	r_k &= \smallint\limits_{{t}_{x[k]}}^{{t}_{x[k]}+\tau }{\left( {{x}^{\intercal}}(t)Qx(t)+u_{i}^{\intercal}(t)R{{u}_{i}}(t) \right)dt}, \label{Equation:r_datum}
	\end{align}
\end{subequations}
where the $k^{th}$ state measurements pair $x_{\tau}[k]$ and $x[k]$ and the reward $r_k$ are such that $r_k$ is integrated over a state-space trajectory segment evolving per \eqref{Equation:StateTrajectories} such that $x_{\tau}[k] = e^{(A+BK_{i} )\tau } x[k]$. Note that unlike \cite{Vrabie_2009}, we do not assume that $x[0],\hdots,x[k],\hdots,x[N-1]$ trace the same state-space trajectory $\varphi(t,x_0,u_{i}(t))$. Note also that if the dependent regressors are included, then the columns of $\mathbf{A}$ will be dependent. Moreover, $N\geq \frac{n(n+1)}{2}$ and the linear system of unknowns \eqref{Equation:Ab_datum} is \emph{consistent}, and $\mathbf{w} = {\mathbf{A}}^{-1}\mathbf{b}$.

\begin{rem} \label{Remark:Collinearity}
	It was suggested in \cite{Vrabie_2009} -- as well as in \cite{MurrayLendaris_ADP_2002} where the special case of letting $\tau =\infty $ in \eqref{Equation:PiIntegralReward} is treated -- that data collection can be along a single state-space trajectory $\varphi(t,x_0,u_{i}(t))$. In that case, the $k^{th}$ row of \eqref{Equation:A_datum} is $\phi(x(t_k+\tau))^{\intercal} - \phi(x(t_k))^{\intercal}$ and $r_k = \int_{{{t}_{k}}}^{{{t}_{k}}+\tau }{\left( {{x}^{\intercal}}(t)Qx(t)+u_{i}^{\intercal}(t)R{{u}_{i}}(t) \right)dt}$ for \eqref{Equation:b_datum}. In Section \ref{Section:NewSolvabilityAnalysis}, we show conditions under which data along a single state-space trajectory is collinear; we suggest an adjusted controller synthesis to avoid data collinearity.
\end{rem}

\subsubsection{Off-Policy}

Given the cost-to-go $V_{i} (x)$ for policy $u_{i} (x)$ as determined by \eqref{Equation:Cost2Go}. Differentiating $V_{i} (x)$ over the trajectories $\varphi(t,x_0,u(t))$ generated by a behavior policy $u(t)=u_b(t)=Fx(t)$ such that 
\begin{equation} \label{Equation:StateTrajectoriesOffPolicy} 
x(t)=e^{(A+BF )(t-t_{0} )} x(t_{0}),
\end{equation} 
we then get
\begin{equation} \label{Equation:Cost2GoOffPolicy} 
\begin{aligned} 
\dot{V}_{i}  &= {\textstyle\frac{dV_{i} }{dx}} ^{\intercal} \left(Ax+Bu\right) \\ 
&= {\textstyle\frac{dV_{i} }{dx}} ^{\intercal} \left(Ax+Bu_{i} \right)+{\textstyle\frac{dV_{i} }{dx}} ^{\intercal} B\Delta(u,u_i) \\ 
&= -x^{\intercal} Qx-u_{i}^{\intercal} Ru_{i} +{\textstyle\frac{dV_{i} }{dx}} ^{\intercal} B\Delta(u,u_i),
\end{aligned} 
\end{equation} 
where $\Delta(u,u_i)=u-u_{i} $. In \cite{AbukhalafLewis_NNHJB2005}, \cref{Equation:Cost2GoOffPolicy} is used to show that $V_{i} $ serves as a Lyapunov function to the dynamics driven by the improved policy $u_{i+1} =-{\textstyle\frac{1}{2}} R^{-1} B^{\intercal} {\textstyle\frac{dV_{i} }{dx}} $, \textit{i.e}. \eqref{Equation:KleinmanUpdate}, namely
\begin{equation*}
\begin{aligned}
\dot{V}_{i}  &= -x^{\intercal} Qx-u_{i}^{\intercal} Ru_{i} +{\textstyle\frac{dV_{i} }{dx}} ^{\intercal} B(u_{i+1} -u_{i} ) \\ 
&= -x^{\intercal} Qx-u_{i}^{\intercal} Ru_{i} -(u_{i+1} -u_{i} )^{\intercal} R(u_{i+1} -u_{i} ).
\end{aligned} 
\end{equation*}

In \cite{Jiang_2012}, \cref{Equation:Cost2GoOffPolicy} is integrated over $\varphi(t,x_0,u(t))$  as follows
\[\smallint\limits_{t_{k} }^{t_{k} +\tau }\dot{V}_{i} dt  = \smallint\limits_{t_{k} }^{t_{k} +\tau }{\textstyle\frac{dV_{i} }{dx}} ^{\intercal} \left(Ax+Bu\right)dt,\]
leading to
\begin{multline} \label{Equation:IntegralRewardOffPolicy}
V_{i} (x(t_{k} +\tau ))-V_{i} (x(t_{k} ))-\smallint\limits_{t_{k} }^{t_{k} +\tau }{\textstyle\frac{dV_{i} }{dx}} ^{\intercal} B\Delta(u,u_i)dt\\ 
=-\smallint\limits _{t_{k} }^{t_{k} +\tau }\left(x^{\intercal} Qx+u_{i}^{\intercal} Ru_{i} \right)dt.
\end{multline}
Thus $V_{i}(x)$ satisfies both \cref{Equation:IntegralReward,Equation:IntegralRewardOffPolicy} where in the former the system is driven by $u(x)=u_b(x)=u_{i}(x) $, and in the latter it is driven by $u(x)=u_b(x)\neq u_i(x)$. In off-policy learning, the behavior policy $u_b$ and thus the closed-loop dynamics can be fixed during target policy iterations.

\begin{rem} \label{Remark:Nonunique}
	In \cite{Jiang_2012}, the policy update relation \eqref{Equation:PolicyUpdate} is embedded in \eqref{Equation:IntegralRewardOffPolicy} to eliminate the explicit knowledge of $B$ requirement and to solve simultaneously for $V_{i} $ and $u_{i+1}$ directly as follows
	\begin{multline} \label{Equation:IntegralRewardOffPolicyNoB} 
	\hat{V}_{i} (x(t_{k} +\tau ))-\hat{V}_{i} (x(t_{k} ))+ 2\smallint\limits _{t_{k} }^{t_{k} +\tau }{\Delta(u,u_i)}^{\intercal} R\hat{u}_{i+1} dt\\ =-\smallint\limits _{t_{k} }^{t_{k} +\tau }\left(x^{\intercal} Qx+u_{i}^{\intercal} Ru_{i} \right)dt.
	\end{multline}
	Reference \cite{Jiang_2012} and subsequent work \cite{Gao_ADPOutputRegulator2016,Gao_ConnectedVehicles2017,Gao_LeaderFormation2019} incorrectly assume that \eqref{Equation:IntegralRewardOffPolicyNoB} has a unique solution pair $\{V_i,u_{i+1}\}$ corresponding to the sequence determined by \eqref{Equation:KleinmanIteration} and \eqref{Equation:KleinmanUpdate}. In Section \ref{Section:NewSolvabilityAnalysis}, we show that \eqref{Equation:IntegralRewardOffPolicy} has a \textit{unique} solution that corresponds to the cost-to-go $V_{i} $ for $u_{i}$, and that  \eqref{Equation:IntegralRewardOffPolicyNoB} has \textit{nonunique} solutions one of them corresponds to  $V_{i}$. Therefore, in our proposed approach to off-policy learning, we require prior knowledge of $B$ and learn using \eqref{Equation:IntegralRewardOffPolicy} not \eqref{Equation:IntegralRewardOffPolicyNoB}.
\end{rem}

Substitute $V_i(x)=x^{\intercal}P_ix$ and $\Delta(u,u_i)=Fx-K_{i}x=L_ix$ in \eqref{Equation:IntegralRewardOffPolicy} to get
\begin{multline} \label{Equation:PiIntegralRewardOffPolicy} 
x(t_{k}+\tau )^{\intercal}{{P}_{i}}x({{t}_{k}}+\tau )-x{{({{t}_{k}})}^{\intercal}}{{P}_{i}}x({{t}_{k}}) \\
-\smallint\limits_{t_{k} }^{t_{k} +\tau }{2x(t)^{\intercal} {P}_{i}} BL_i x(t)dt\\
=-\smallint\limits_{{{t}_{k}}}^{{{t}_{k}}+\tau }{\left( {{x}^{\intercal}}(t)Qx(t)+u_{i}^{\intercal}(t)R{{u}_{i}}(t) \right)dt}.
\end{multline}
The result is a response-based data-driven approach similar to \eqref{Equation:PiIntegralReward} to solve for $P_i$ given online measurements. 

Using the notation of \Cref{Section:Notation}, we use a quadratic form $V(x) = {\mathbf{w}}^{\intercal}\phi(x)$, where $\phi(x)=vec^{L} (xx^{\intercal} )$, and $\textstyle\frac{dV}{dx} = \nabla\phi(x)^{\intercal} {\mathbf{w}}$ to solve for $\mathbf{w}$ from measurements by forming the following linear system of unknowns $\mathbf{A}{\mathbf{w}} = \mathbf{b}$; where $\mathbf{A} \in \mathbb{R}^{N\times \frac{n(n+1)}{2}}$ and $\mathbf{b} \in \mathbb{R}^{N\times 1}$ are such that
\begin{subequations} \label{Equation:Ab_datumOffPolicy}
	\begin{align}
	\mathbf{A} &=  \left[\begin{array}{l} \label{Equation:A_datumOffPolicy}
	\phi(x_{\tau}[0])^{\intercal} - \phi(x[0])^{\intercal} - \delta_{0}\\ 
	\vdots\\
	\phi(x_{\tau}[k])^{\intercal} - \phi(x[k])^{\intercal} - \delta_{k}\\
	\vdots\\
	\phi(x_{\tau}[N-1])^{\intercal} - \phi(x[N-1])^{\intercal} - \delta_{N-1}
	\end{array}\right],\\
	\mathbf{b} &=  -\left[\begin{array}{l}  \label{Equation:b_datumOffPolicy}
	r_0, \hdots, r_k, \hdots, r_{N-1}
	\end{array}\right]^{\intercal},\\
	\delta_{k} &= \smallint\limits_{{t}_{x[k]}}^{{t}_{x[k]}+\tau }{x(t)^{\intercal}L_{i}^{\intercal}B^{\intercal} \nabla\phi(x(t))^{\intercal}} dt, \label{Equation:delta_datumOffPolicy} \\
	r_k &= \smallint\limits_{{t}_{x[k]}}^{{t}_{x[k]}+\tau }{\left( {{x}^{\intercal}}(t)Qx(t)+u_{i}^{\intercal}(t)R{{u}_{i}}(t) \right)dt}, \label{Equation:r_datumOffPolicy}
	\end{align}
\end{subequations}
where the $k^{th}$ state measurements pair $x_{\tau}[k]$ and $x[k]$ and the reward $r_k$ are such that $r_k$ is integrated over a state-space trajectory segment evolving per \eqref{Equation:StateTrajectoriesOffPolicy} such that $x_{\tau}[k] = e^{(A+BF )\tau } x[k]$. We do not assume that $x[0],\hdots,x[k],\hdots,x[N-1]$ trace the same state-space trajectory $\varphi(t,x_0,u(t))$. Moreover, $N\geq \frac{n(n+1)}{2}$ and the linear system of unknowns \eqref{Equation:Ab_datumOffPolicy} is \emph{consistent}.

In off-policy iterations, $x(t)$ evolves per the behavior policy $u_b=Fx$ and is independent of $u_i$. Therefore, closed-loop data collected at iteration $i=0$ can be used in subsequent iterations, \emph{experience replay}, by recomputing $K_i$ dependent terms in \eqref{Equation:Ab_datumOffPolicy}. To do so, re-write \eqref{Equation:delta_datumOffPolicy} and \eqref{Equation:r_datumOffPolicy} by factoring out $K_i$ using $K_i = \sum_{p=1}^{n} \sum_{q=1}^{m} [K_{i}]_{(p,q)} \mathbf{e}_{pq}$ such that,
\begin{subequations} \label{Equation:deltaBreak}
	\begin{align}
	\delta_{k} =& \delta_k(F) - \delta_{k}(K_i), \label{Equation:deltaAll}\\
	\delta_{k}(F) =& \smallint\limits_{{t}_{x[k]}}^{{t}_{x[k]}+\tau }{u(x(t))^{\intercal}B^{\intercal} \nabla\phi(x(t))^{\intercal}} dt, \label{Equation:deltaF}\\
	\delta_k(\mathbf{e}_{pq}) =& \smallint\limits_{{t}_{x[k]}}^{{t}_{x[k]}+\tau }{x(t)^{\intercal}\mathbf{e}_{pq}^{\intercal}B^{\intercal} \nabla\phi(x(t))^{\intercal}} dt, \label{Equation:deltae}\\
	\delta_{k}(K_i) =& \sum_{p=1}^{n} \sum_{q=1}^{m} [K_{i}]_{(p,q)} \delta_k (\mathbf{e}_{pq}), \label{Equation:deltaKi}
	\end{align}
\end{subequations}

\begin{subequations} \label{Equation:rBreak}
	\begin{align}
	r_{k} =& r_{k}(Q) + r_{k}(K_i), \label{Equation:rAll}\\
	r_{k}(Q) =&  \smallint\limits_{{t}_{x[k]}}^{{t}_{x[k]}+\tau }{ x(t)^{\intercal} Q x(t)  dt}, \label{Equation:rQ}\\
	r(\mathbf{e}_{pq}) =& \smallint\limits_{{t}_{x[k]}}^{{t}_{x[k]}+\tau }{ x(t)^{\intercal} \mathbf{e}_{pq} x(t)  dt},\; \mathbf{e}_{pq} \in \mathbb{R}^{n \times n}, \label{Equation:re}\\
	r_{k}(K_i) =&  \sum_{p=1}^{n} \sum_{q=1}^{n} [K_i^{\intercal}RK_i]_{(p,q)} r (\mathbf{e}_{pq}). \label{Equation:rKi}
	\end{align}
\end{subequations}

The $K_i$ dependent terms are \eqref{Equation:deltaKi} and \eqref{Equation:rKi}. Note that $\mathbf{e}_{ij}$ is defined in \Cref{Section:Notation}. If ${\mathbf{A}}$ in \eqref{Equation:A_datumOffPolicy} remains invertible for the recomputed $K_i$ dependent terms, then there is no need for new closed-loop data. Otherwise, new data is needed such that the choice of $x[k]$ measurements ensures ${\mathbf{A}}$ in \eqref{Equation:A_datumOffPolicy} is invertible. Note that in \eqref{Equation:deltaF}, $u(x)=Fx$ and thus knowledge of $F$ is not required as long as $u(x)$ is measurable.

\section{Policy Iterations: New Solvability Analysis} \label{Section:NewSolvabilityAnalysis}

In this section, we address the issues raised by \Cref{Remark:Collinearity} and \Cref{Remark:Nonunique}. \Cref{Theorem:UniqueNonunique} addresses \Cref{Remark:Nonunique}.

\begin{thm} \label{Theorem:UniqueNonunique}
	Let $K_{i}$ be fixed and such that $A+BK_{i} $ is Hurwitz. Assume $u(x)=Fx$ drives \eqref{Equation:LinearSystem} such that $x(t)=e^{(A+BF)(t-t_{k})} x(t_{k} )$. Let $\Delta(u,u_i)=Fx-K_{i}x=L_ix$. It follows that:
	\begin{enumerate}	
		\item[A.]  $\forall F$, there exists a \emph{unique} solution $\hat{V}(x)$ to the integral equation
		\begin{subequations} \label{Equation:UniqueIntegroDifferentialSystem}
			\begin{align}
			\begin{split} \label{Equation:UniqueIntegroDifferential}
			\hat{V}(x({{t}_{k}}+\tau ))&-\hat{V}(x({{t}_{k}}))-\smallint\limits_{{{t}_{k}}}^{{{t}_{k}}+\tau }{{{\Delta(u,u_i)}^{\intercal}}{{B}^{\intercal}}\tfrac{d\hat{V}}{dx}dt}\\
			&=-\smallint\limits_{{{t}_{k}}}^{{{t}_{k}}+\tau }{\left( {{x}^{\intercal}}Qx+u_{i}^{\intercal}R{{u}_{i}} \right)dt},
			\end{split}\\
			\hat{V}(0)&=0, \label{Equation:UniqueIntegroDifferentialBoundary}
			\end{align}
		\end{subequations} 
		with integration terms carried over $\varphi(t,t_{k} ,x(t_{k} ),u=Fx)$. The solution is given by $\hat{V}(x)=x^{\intercal} P_{i} x$, where $P_{i} $ is defined by \eqref{Equation:Cost2GoMatrix} and is the unique solution to \eqref{Equation:KleinmanIteration} for the associated $K_{i} $.
		\item[B.]  $\forall F$, there exists a \emph{nonunique} solution pair $\{\hat{V}(x),\hat{u}(x)\}$ to the integral equation
		\begin{subequations} \label{Equation:NonUniqueIntegroDifferentialSystem}
			\begin{align}
			\begin{split} \label{Equation:NonUniqueIntegroDifferential} 
			\hat{V}(x({{t}_{k}}+\tau ))&-\hat{V}(x({{t}_{k}}))+2\smallint\limits_{{{t}_{k}}}^{{{t}_{k}}+\tau }{{{\Delta(u,u_i)}^{\intercal}}R\hat{u}dt}\\
			&=-\smallint\limits_{{{t}_{k}}}^{{{t}_{k}}+\tau }{\left( {{x}^{\intercal}}Qx+u_{i}^{\intercal}R{{u}_{i}} \right)dt},
			\end{split}\\
			\hat{V}(0)&=0,\ \hat{u}(0)=0. \label{Equation:NonUniqueIntegroDifferentialBoundary}
			\end{align}			
		\end{subequations}
	\end{enumerate}
\end{thm}

\begin{proof}
	Existence of solutions to \eqref{Equation:UniqueIntegroDifferentialSystem}, follows from the fact that $\hat{V}(x)=x^{\intercal} P_ix$ is a solution as it an be differentiated over $\varphi(t,t_{k} ,x(t_{k} ),u=Fx)$ as in \eqref{Equation:Cost2GoOffPolicy} then integrated as in \eqref{Equation:IntegralRewardOffPolicy}, where $P_i$ is defined by \eqref{Equation:Cost2GoMatrix}. To show uniqueness, assume there is a solution given by $\hat{V}(x)=x^{\intercal} \hat{P}x$ and substitute it together with $x(t_{k} +\tau )=e^{(A+BF)\tau } x(t_{k} )$ in \eqref{Equation:UniqueIntegroDifferential} to get
	\begin{multline} \label{Equation:SolvabiliyUniqueIntegroDifferential} 
	x{{({{t}_{k}})}^{\intercal}}{{\left( {{\operatorname{e}}^{A_F\tau }} \right)}^{\intercal}}\hat{P}{{\operatorname{e}}^{A_F\tau }}x({{t}_{k}})-x{{({{t}_{k}})}^{\intercal}}\hat{P}x({{t}_{k}}) \\ 
	-2x{{({{t}_{k}})}^{\intercal}}\smallint\limits_{{{t}_{k}}}^{{{t}_{k}}+\tau }{{{\left( {{\operatorname{e}}^{A_F(t-{{t}_{k}})}} \right)}^{\intercal}}{{L}_{i}^{\intercal}}{{B}^{\intercal}}\hat{P}{{\operatorname{e}}^{A_F(t-{{t}_{k}})}}dt}x({{t}_{k}})= \\ 
	-x{{({{t}_{k}})}^{\intercal}}\smallint\limits_{{{t}_{k}}}^{{{t}_{k}}+\tau }{ {{\left( {{\operatorname{e}}^{A_F(t-{{t}_{k}})}} \right)}^{\intercal}}\left[ K_{i}^{\intercal}R{{K}_{i}}+Q \right]{{\operatorname{e}}^{A_F(t-{{t}_{k}})}} dt} x({{t}_{k}}),
	\end{multline}
	
	\noindent where $L=F-K_{i} $ and $A_F = A+BF$. Since \eqref{Equation:SolvabiliyUniqueIntegroDifferential} is valid for all $x(t_{k} )$, we have
	\begin{multline} \label{Equation:SolvabiliyUniqueIntegroDifferentialNox} 
	{{\left( {{\operatorname{e}}^{A_F\tau }} \right)}^{\intercal}}\hat{P}{{\operatorname{e}}^{A_F\tau }}-\hat{P} \\ 
	-\smallint\limits_{{{t}_{k}}}^{{{t}_{k}}+\tau }{{{\left( {{\operatorname{e}}^{A_F(t-{{t}_{k}})}} \right)}^{\intercal}}\left[ {{L}_{i}^{\intercal}}{{B}^{\intercal}}\hat{P}+\hat{P}B{L}_{i} \right]{{\operatorname{e}}^{A_F(t-{{t}_{k}})}}dt} \\ 
	=-\smallint\limits_{{{t}_{k}}}^{{{t}_{k}}+\tau }{ {{\left( {{\operatorname{e}}^{A_F(t-{{t}_{k}})}} \right)}^{\intercal}}\left[ K_{i}^{\intercal}R{{K}_{i}}+Q \right]{{\operatorname{e}}^{A_F(t-{{t}_{k}})}} dt}.
	\end{multline}
	
	Since the right-hand side and left-hand side of \eqref{Equation:SolvabiliyUniqueIntegroDifferentialNox} are smooth and analytic, their Taylor series expansions must be equal. Differentiating \eqref{Equation:SolvabiliyUniqueIntegroDifferentialNox} once with respect to $\tau $ we get the first term of the Taylor series,
	\begin{multline} \label{Equation:FirstTermTaylor}
	{{\left( {{\operatorname{e}}^{(A+BF)\tau }} \right)}^{\intercal}}{{(A+BF)}^{\intercal}}\hat{P}{{\operatorname{e}}^{(A+BF)\tau }}\\
	+{{\left( {{\operatorname{e}}^{(A+BF)\tau }} \right)}^{\intercal}}\hat{P}(A+BF){{\operatorname{e}}^{(A+BF)\tau }} \\ 
	-{{\left( {{\operatorname{e}}^{(A+BF)\tau }} \right)}^{\intercal}}\left[ {{L}_{i}^{\intercal}}{{B}^{\intercal}}\hat{P}+\hat{P}B{L}_{i} \right]{{\operatorname{e}}^{(A+BF)\tau }} \\ 
	=- {{\left( {{\operatorname{e}}^{(A+BF)\tau }} \right)}^{\intercal}}\left[ K_{i}^{\intercal}R{{K}_{i}}+Q \right]{{\operatorname{e}}^{(A+BF)\tau }}.
	\end{multline} 
	\noindent Setting $\tau =0$ in \eqref{Equation:FirstTermTaylor}, this shows that $\hat{P}$ must satisfy
	\begin{multline} \label{Equation:FirstTermTaylorCoef} 
	(A+BF)^{\intercal} \hat{P}+\hat{P}(A+BF)-{L}_{i}^{\intercal} B^{\intercal} \hat{P}-\hat{P}B{L}_{i}\\=-K_{i}^{\intercal} RK_{i} -Q.
	\end{multline} 
	
	\noindent Substitute $L_i=F - K_{i}$ in \eqref{Equation:FirstTermTaylorCoef} and canceling common terms, it follows that $\hat{P}$ must be a solution to $\hat{P} (A+BK_{i} )+(A+BK_{i} )^{\intercal} \hat{P} +K_{i}^{\intercal} RK_{i} +Q=0$ which due to the Hurwitz condition on $A+BK_{i}$ must have a unique solution which is $P_i$ per \eqref{Equation:KleinmanIteration}. Hence, $\hat{P} = P_i$.
	
	Existence of solutions to \labelcref{Equation:NonUniqueIntegroDifferentialSystem} follows directly from substituting $\hat{u}(x)=-R^{-1} B^{\intercal} P_{i}$ in \labelcref{Equation:NonUniqueIntegroDifferential} to get \labelcref{Equation:UniqueIntegroDifferential} and then substituting $\hat{V}(x)=x^{\intercal} P_ix$, thus $\{\hat{V}(x)=x^{\intercal} P_ix, \hat{u}(x)=-R^{-1} B^{\intercal} P_{i}\}$ is a solution pair. It remains to show nonuniqueness which we accomplish by constructing other solutions. Assume there is a solution given by $\hat{V}(x)=x^{\intercal} \hat{P}x$ and $\hat{u}(x)=\hat{K}x$. Substitute it together with $x(t_{k} +\tau )=e^{(A+BF)\tau } x(t_{k} )$ in \eqref{Equation:NonUniqueIntegroDifferential} to get
	\begin{multline} \label{Equation:SolvabiliyNonUniqueIntegroDifferential} 
	x{{({{t}_{k}})}^{\intercal}}{{\left( {{\operatorname{e}}^{A_F\tau }} \right)}^{\intercal}}\hat{P}{{\operatorname{e}}^{A_F\tau }}x({{t}_{k}})-x{{({{t}_{k}})}^{\intercal}}\hat{P}x({{t}_{k}}) \\ 
	+2x{{({{t}_{k}})}^{\intercal}}\smallint\limits_{{{t}_{k}}}^{{{t}_{k}}+\tau }{{{\left( {{\operatorname{e}}^{A_F(t-{{t}_{k}})}} \right)}^{\intercal}}{L}_{i}^{\intercal} R\hat{K}{{\operatorname{e}}^{A_F(t-{{t}_{k}})}}dt}x({{t}_{k}}) =-\\ 
	x{{({{t}_{k}})}^{\intercal}}\smallint\limits_{{{t}_{k}}}^{{{t}_{k}}+\tau }{ {{\left( {{\operatorname{e}}^{A_F(t-{{t}_{k}})}} \right)}^{\intercal}}\left[ K_{i}^{\intercal}R{{K}_{i}}+Q \right]{{\operatorname{e}}^{A_F(t-{{t}_{k}})}} dt} x({{t}_{k}}),
	\end{multline} 
	
	\noindent where $L_i=F-K_{i} $ and $A_F = A+BF$. Similar to the steps done for \labelcref{Equation:SolvabiliyUniqueIntegroDifferential,Equation:SolvabiliyUniqueIntegroDifferentialNox}, $\hat{P}$ and $\hat{K}$ must satisfy
	\begin{multline} \label{Equation:FirstTermTaylorCoefNonUnique} 
	(A+BF)^{\intercal} \hat{P}+\hat{P}(A+BF)+{L}_{i}^{\intercal} R\hat{K}+\hat{K}^{\intercal} R{L}_{i}\\=-K_{i}^{\intercal} RK_{i} -Q.
	\end{multline}
	Conversely, if $\{\hat{P},\hat{K}\}$ is a solution to \eqref{Equation:FirstTermTaylorCoefNonUnique}, then $\{\hat{V}(x)=x^{\intercal} \hat{P}x,\hat{u}(x)=\hat{K}x\}$ is a solution to \eqref{Equation:NonUniqueIntegroDifferential}. To see this, note that \eqref{Equation:FirstTermTaylorCoefNonUnique} can be written as $\dot{\hat{V}} + 2\textstyle\hat{u}^{\intercal}R\Delta(u,u_i) = -x^{\intercal} Qx-u_{i}^{\intercal} Ru_{i}$, where the time derivative is over $\varphi(t,t_{k} ,x(t_{k} ),u=Fx)$, which can then be integrated from $t_k$ to $t_k+\tau$ over the trajectories $\varphi(t,t_{k} ,x(t_{k} ),u=Fx)$ to get \eqref{Equation:NonUniqueIntegroDifferentialSystem}. Therefore, it is sufficient to show that \eqref{Equation:FirstTermTaylorCoefNonUnique} has nonunique solutions.
	
	Equation \eqref{Equation:FirstTermTaylorCoefNonUnique} can be decomposed into the following
	\begin{subequations} \label{Equation:Sylvesters}
		\begin{align} 
		(A+BF)^{\intercal} \hat{P}+\hat{P}(A+BF)&=W_{1} , \label{Equation:Sylvester}   \\
		{L}_{i}^{\intercal} R\hat{K}+\hat{K}^{\intercal} R{L}_{i}&=W_{2} , \label{Equation:SylvesterTranspose} \\
		-W_{1} -K_{i}^{\intercal} RK_{i} -Q &= W_{2}\label{Equation:CouplingSylvester},
		\end{align}
	\end{subequations}
	
	\noindent where \eqref{Equation:Sylvester} is a Sylvester equation to solve for $\hat{P}\in R^{n\times n} $ and \eqref{Equation:SylvesterTranspose} is a Sylvester-transpose equation to solve for $\hat{K}\in R^{m\times n} $; both Sylvester equations are coupled by \eqref{Equation:CouplingSylvester}. The Sylvester equations correspond to linear maps $\mathcal{T}_{1} :\hat{P}\to W_1$ and $\mathcal{T}_{2} :\hat{K}\to W_2$, thus $W_1\in \text{Im} \,\mathcal{T}_{1}$ and $W_2\in \text{Im} \,\mathcal{T}_{2}$. Existence of solutions to the Sylvester-transpose \eqref{Equation:SylvesterTranspose} is discussed in \cite{Wimmer_RothTheorem1994,WuZhang_ComplexConjugateMatrixEquations2017}. To construct solutions to \eqref{Equation:FirstTermTaylorCoefNonUnique} from \eqref{Equation:Sylvesters}, first choose $W_1\in \text{Im} \,\mathcal{T}_{1}$ such that $W_2\in \text{Im} \,\mathcal{T}_{2}$ via \eqref{Equation:CouplingSylvester}; then, for the chosen $W_1$ and $W_2$ solve \eqref{Equation:Sylvester} and \eqref{Equation:SylvesterTranspose} separately to find a solution pair $\{\hat{P},\hat{K}\}$. Conversely, given a solution $\{\hat{P},\hat{K}\}$ to \eqref{Equation:FirstTermTaylorCoefNonUnique}, then $\{\hat{P},\hat{K}\}$ satisfies \eqref{Equation:Sylvesters} for an appropriate $W_1$ and $W_2$.
	
	Note that $\forall F$, $\{\hat{P}=P_{i} ,\hat{K}=-R^{-1} B^{\intercal} P_{i} \}$ is a solution to \eqref{Equation:FirstTermTaylorCoefNonUnique}, thus to \labelcref{Equation:Sylvesters}, resulting in $W_{1} ={L}_{i}^{\intercal} B^{\intercal} P_{i} +P_{i} B{L}_{i}-K_{i}^{\intercal} RK_{i} -Q$ and $W_{2} =-{L}_{i}^{\intercal} B^{\intercal} P_{i} -P_{i} B{L}_{i}$ where $L_i=F-K_{i} $, and where both $W_1\in \text{Im} \,\mathcal{T}_{1}$ and $W_2\in \text{Im} \,\mathcal{T}_{2}$. We can construct additional solutions to \labelcref{Equation:Sylvesters}, and thus to \eqref{Equation:FirstTermTaylorCoefNonUnique}, as follows:
	
	\paragraph{Common eigenvalue} If $A+BF$ and $-(A+BF)$ have a common eigenvalue $\lambda $, choose $W_{1} =L^{\intercal} B^{\intercal} P_{i} +P_{i} BL-K_{i}^{\intercal} RK_{i} -Q$ and thus $W_{2} =-{L}_{i}^{\intercal} B^{\intercal} P_{i} -P_{i} B{L}_{i}$ which as noted earlier are $W_1\in \text{Im} \,\mathcal{T}_{1}$ and $W_2\in \text{Im} \,\mathcal{T}_{2}$. Clearly, $\hat{K}=-R^{-1} B^{\intercal} P_{i}$ solves \eqref{Equation:SylvesterTranspose}. Let $w^{\intercal} (A+BF)=-\lambda w^{\intercal} $ and $(A+BF)^{\intercal} v=\lambda v$. It follows that $\hat{P}=P_{i} +vw^{\intercal} $ is a solution to \eqref{Equation:Sylvester} where $vw^{\intercal} \in \ker \mathcal{T}_{1}$. Thus $\{\hat{P}=P_{i} +vw^{\intercal} ,\hat{K}=-R^{-1} B^{\intercal} P_{i}\}$ is a solution to \eqref{Equation:FirstTermTaylorCoefNonUnique}. 
	\paragraph{No common eigenvalue} If $A+BF$ and $-(A+BF)$ have no common eigenvalue, then $\forall W_{1} $ there exists a unique solution to \eqref{Equation:Sylvester}. Let $W_{1} =-K_{i}^{\intercal} RK_{i} -Q$ and let $P_{W_{1} } $ be the associated unique solution to \eqref{Equation:Sylvester}. If for example $A+BF$ is Hurwitz, then 
	\begin{equation*}
	P_{W_{1}} = \smallint\limits_{t_{0} }^{\infty } \left[e^{(A+BF )(t-t_{0} )} \right]^{\intercal} \left[Q+K_{i}^{\intercal} RK_{i} \right]e^{(A+BF )(t-t_{0} )}  dt.
	\end{equation*}
	
	\noindent From $W_{1} $, it follows that $W_{2} =0$ and thus a solution to \eqref{Equation:SylvesterTranspose} would be such that $\hat{K}\in \ker \mathcal{T}_{2}$. Thus $\{\hat{P}=P_{W_{1} } ,\hat{K}=0\}$ is a solution to \eqref{Equation:Sylvesters}, and thus to \eqref{Equation:FirstTermTaylorCoefNonUnique}.
	
\end{proof}

The following results are in relation to the choice of data $x(t_k)$ and $x(t_k+\tau)$ used to solve the linear system \eqref{Equation:Ab_datum}. In particular we address the data collinearity issue raised in \Cref{Remark:Collinearity} and analyze the data-driven computational scheme to dissect the role of exploitation from that of exploration.

\begin{defn} \label{Definition:Spectrum}
	The \emph{spectrum} of a square matrix $A$ is
	\[\Lambda(A) \buildrel\Delta\over= \{\forall \lambda: \det(\lambda I - A)=0 \}.\]
\end{defn}

\begin{defn} \label{Definition:LinearIndependence}
	Let $\sigma(x) = [\sigma_1(x),\hdots,\sigma_N(x)]^{\intercal}$ where $\sigma_i(x):\mathbb{R}^n \rightarrow \mathbb{R}$, and $N\geq 2$. Let  $\Gamma=\{a \in \mathbb{R}^{N \times 1}: a_i \neq 0, a_j \neq 0, i \neq j \}$. The set of $N$ functions $\sigma_i(x)$ is \emph{dependent} iff
	\[\exists w \in \Gamma, \forall x \in \mathbb{R}^n: \sigma(x)^{\intercal} w = 0,\]
	and is \emph{independent} iff 
	\[\forall w \in \Gamma, \exists x \in \mathbb{R}^n: \sigma(x)^{\intercal} w \neq 0.\]
	
\end{defn}

\begin{lem} \label{Lemma:UnitVectors}
	$\exists X=[x(1),\hdots,x(N)] \in \mathbb{R}^{n\times N}$ such that $\Phi=[\phi(x(1)),\hdots,\phi(x(N))] \in \mathbb{R}^{N\times N}$ is full rank, where $\phi(x(\cdot)) = vec^{L}(x(\cdot)x(\cdot)^{\intercal})$ and $N = \frac{n(n+1)}{2}$.
\end{lem}

\begin{cor} \label{Corollary:LinearIndependence}
	Let $\phi(x) = vec^{L}(xx^{\intercal})$ where $x \in \mathbb{R}^n$. The set of $N= \frac{n(n+1)}{2}$ functions $\phi_i(x)$ is linearly independent.	
\end{cor}

\begin{lem} \label{Lemma:ExistenceData}
	Let $N = \frac{n(n+1)}{2}$, and let $\mathbf{A} = [\sigma(x(t_k)) \hdots \sigma(x(t_N))]^{\intercal}$ where $\sigma(x) = [\sigma_1(x),\hdots,\sigma_N(x)]^{\intercal}$ with $\sigma(x(t_k)) = \phi(x(t_k)) - \phi(x(t_k+\tau))$ as shown in \eqref{Equation:A_datum}. Let $x(t_k+\tau)$ satisfy \eqref{Equation:StateTrajectories} for a stabilizing $K_i$. If $\phi_1(x),\hdots,\phi_N(x)$ is a linearly independent set, then $\forall \tau: \exists [x(t_1),\hdots, x(t_N)]:rank(\mathbf{A})=N$.
\end{lem}

To explore the cases for which on-policy learning along a single state-space trajectory fails requires decomposing \eqref{Equation:StateTrajectories} in terms of its generalized modes. Similar analysis is possible for off-policy learning as well. Due to space limitation, we limit the discussion here to the following lemma and theorem useful for \Cref{Section:Example}.
\begin{lem} \label{Lemma:DiagonalizableDegenerateNonDegenerate}
	Let $A$ be diagonalizable. Assume $\exists \lambda \in \Lambda(A)$ such that $\lambda$ is $m$-fold degenerate. $\forall x(t_0)  \in \mathbb{R}^{n}$, if $\mathbf{A}$ in \eqref{Equation:A_datum} is formed from data points along $x(t_k)=e^{A(t_k-t_{0} )} x(t_{0})$, then $rank(\mathbf{A})<N$.
\end{lem}

\begin{thm} \label{Theorem:MinimumNumberUniqueTrajectories}
	To avoid data collinearity in \eqref{Equation:A_datum}, it is necessary to explore a minimum number of $\frac{n(n+1)}{2}$ pairwise distinct state-space trajectories.
\end{thm}


\section{Shared Linear Quadratic Regulator} \label{Section:sLQR}

In this section, we leverage the analysis and insights provided in \Cref{Section:PreliminariesPolicyIterations,Section:NewSolvabilityAnalysis} to synthesize solutions to \Cref{Problem:MainProblem,Problem:TakeOver} posed in \Cref{Section:ProblemFormulation}.

In \Cref{Problem:MainProblem}, the aim is to find $u_a(x)$ that optimizes the human-in-the-loop closed-loop dynamics \eqref{Equation:HCI} by minimizing \eqref{Equation:MainProblem}. Since the dynamics as seen by $u_a(x)$ is shaped by the human input, we lump the human input with $A$ and rewrite the closed-loop dynamics as follows
\begin{equation} \label{Equation:SynthesisProblem1} 
\dot{x}=A_h+Bu_a(x),
\end{equation}
\noindent where ${{A}_{h}}=A+B{K_hC_h}$ is unknown to the autonomy system. Note that the integrand in \eqref{Equation:MainProblem} is quadratic in $x$ and $u_a$, thus the underlying ARE is given by
\begin{equation} \label{Equation:ARE-ParallelAutonomy} 
0=PA_{h}+{A_h}^{\intercal}P- PBR^{-1}B^{\intercal}P+Q_h,
\end{equation}
where $Q_h=Q+{C}_{h}^{\intercal}K_{h}^{\intercal}M{{K}_{h}{C}_{h}}$. The minimum autonomy intervention policy is given by $u_a(x)=-R^{-1} B^{\intercal}P x$ where $P$ is the stabilizing solution of \eqref{Equation:ARE-ParallelAutonomy}.

To solve \eqref{Equation:ARE-ParallelAutonomy} in a data-driven way, we can use either on-policy learning or off-policy learning. In both cases, we require that $A+B{K_hC_h}$ is Hurwitz. Moreover, the cost function's design parameters $M, Q, R$, the learning data size $N$, and the duration of the reward window $\tau$ are all required. Finally, access to signals $x(t)$ and $u_h(x)$ as well as knowledge of the input matrix $B$ are all required by the autonomy system.

In on-policy learning of a minimum intervention policy, we let $u_a(x) = u_i(x)$ at each policy iteration thus the closed-loop is changing at each iteration. We initialize $u_0(x) = 0$ since the open-loop is already stable due to $A_{h}$ being Hurwitz. This is shown in \Cref{Algorithm:InterventionOnPolicy}.

In  off-policy learning of a minimum intervention policy, we let $u_a(x) = 0$ at each policy iteration thus the closed-loop is fixed -- the behavior policy is $u_a$ and thus ${\Delta(u_a,u_i)}$ is used in $\delta_k$ in \eqref{Equation:deltaAll}. We initialize $u_0(x) = 0$ since $A_{h}$ is Hurwitz. For brevity, we do not show the full algorithm, but it is along the same lines of the off-policy algorithm shown in \Cref{Algorithm:TakeOverOffPolicy} for the takeover problem. You can also consult \cite{CodeMATLAB} for MATLAB implementation. Note that target policy $u_i$ is floating during iterations, \emph{i.e.} never applied to the plant.

\begin{algorithm}[H]
	\caption{\small Learning Minimum Intervention Policy (\textbf{On-Policy})}\label{Algorithm:InterventionOnPolicy}
	\begin{algorithmic}[1]		
		\Function{Main}{}
		\State \Call{Initialize}{$\mathbf{0}$,$t_0$,$x_0$}
		\Repeat
		\State $u_a(x(t))\gets u_i(x(t))$ \Comment{Closed-loop updates $\forall i$}
		\State $\mathbf{A},\mathbf{b} \gets$ \Call{ExploitPolicy}{} \Comment{Exploit to gather data} 
		\State $W \gets {\mathbf{A}}^{-1} \mathbf{b}$; $P_i \gets reshape(W)$ \Comment {Compute weights}
		\State $K_{i+1} \gets -R^{-1} B^{\intercal} P_{i}$; $i \gets i+1$; $u_i(x(t))\gets K_ix(t)$
		\Until{$P_i$ converges}
		\State \textbf{return} $P_i$
		\EndFunction
		\Function{Initialize}{gain,time,state} \label{Procedure:Intialize} \Comment{Set $u_0$}
		\State $K_0 \gets gain$; $u_0\gets K_0x(t)$; $t_0 \gets time$; $x(t_0) \gets state$; $i \gets 0$
		\State Prepare $\mathbf{A}$ and $\mathbf{b}$ in \eqref{Equation:Ab_datum} for data.
		\EndFunction
		\Function{ExploitPolicy}{} \label{Procedure:ExploitPolicy}
		\While{$(size(\mathbf{A}) < N) \; \lor (cond(\mathbf{A})<\epsilon)  $}
		\State $r_x,r_{u_{h}},r_{u_{i}} \gets$ \Call{EvaluateReward}{}
		\State Add $r_x,r_{u_{h}},r_{u_{i}}$ to $\mathbf{b}$ in \eqref{Equation:b_datum}; 
		\State Add $\phi(x(t_k)), \phi(x(t_k+\tau))$ data to $\mathbf{A}$ in \eqref{Equation:A_datum}
		\State \Call{Nudge}{} \Comment{To switch to a new trajectory to explore}
		\State  $u_a(x) \gets u_i(x(t))$ \Comment{Continue exploiting this policy}
		\EndWhile
		\State \textbf{return} $\mathbf{A},\mathbf{b}$
		\EndFunction
		\Function{EvaluateReward}{}
		\Statex Dynamics is evolving per $x(t)=e^{(A_h+BK_{i} )(t-t_{k} )} x(t_{k})$.
		\State  $r_x \gets  \smallint\limits_{t_{k}}^{t_{k}+\tau }{ {{x}^{\intercal}}Qxdt}$; $r_{u_{h}} \gets \smallint\limits_{t_{k}}^{t_{k}+\tau }{ u_{h}^{\intercal} (x)Mu_{h}(x) dt}$
		\State  $r_{u_{i}} \gets \smallint\limits_{t_{k}}^{t_{k}+\tau }{ u_{i}^{\intercal}(x)R{{u}_{i}}(x) ) dt}$
		\State \textbf{return} $r_x,r_{u_{h}},r_{u_{i}}$ \Comment{Return reward}
		\EndFunction
		\Function{Nudge}{}
		\State  $u_a(x) \gets u_i(x(t)) + PRNG$ \Comment{Pseudorandom Number Generator}
		\EndFunction
	\end{algorithmic}
\end{algorithm}

\begin{algorithm}[H]
	\caption{\small Learning Optimal Takeover Policy (\textbf{Off-Policy})}\label{Algorithm:TakeOverOffPolicy}
	\begin{algorithmic}[1]
		\Function{Main}{}
		\State \Call{Initialize}{$u_h(x)$,$t_0$,$x_0$}
		\State $u_a(x(t))\gets \mathbf{0}$ \Comment{Closed-loop fixed $\forall i$}
		\State $\mathbf{A},\mathbf{b} \gets$ \Call{ExploitPolicy}{} \Comment{Exploit to gather data} 
		\Repeat
		\State $W \gets {\mathbf{A}}^{-1} \mathbf{b}$; $P_i \gets reshape(W)$ \Comment {Compute weights}
		\State $K_{i+1} \gets -R^{-1} B^{\intercal} P_{i}$; $i \gets i+1$; ; $u_i(x(t))\gets K_ix(t)$
		\State $\mathbf{A},\mathbf{b} \gets$ \Call{Recompute}{$\mathbf{A}$,$\mathbf{b}$} 
		\Until{$P_i$ converges}
		\State \textbf{return} $P_i$
		\EndFunction
		\Function{Initialize}{gain or signal,$time$,$state$} \label{Procedure:OffPolicyIntializeTakeOver} \Comment{Set $u_0$}
		\State $u_0(x(t)) \gets u_h(x(t))$; $t_0 \gets time$; $x(t_0) \gets state$; $i \gets 0$
		\State Prepare $\mathbf{A}$ and $\mathbf{b}$ in \eqref{Equation:Ab_datumOffPolicy} for data.
		\EndFunction
		

		\Function{ExploitPolicy}{} \label{Procedure:ExploitPolicyOff}
		\While{$(size(\mathbf{A}) < N) \; \lor (cond(\mathbf{A})<\epsilon)  $}
		\State $\delta_k(F),\delta_k(\mathbf{e}_{pq}),r_k(Q), r_k(\mathbf{e}_{pq})  \gets$ \Call{EvaluateReward}{}
		\State $\delta_k(K_i) \gets \eqref{Equation:deltaKi}, \delta_k \gets \eqref{Equation:deltaAll}$; $r_k(K_i) \gets \eqref{Equation:rKi},r_k \gets \eqref{Equation:rAll}$
		\State Add $,r_k$ to $\mathbf{b}$ in \eqref{Equation:b_datumOffPolicy}
		\State Add $\phi(x(t_k)), \phi(x(t_k+\tau))$, $\delta_k$ to $\mathbf{A}$ in \eqref{Equation:A_datumOffPolicy}
		\State \Call{Nudge}{} \Comment{To switch to a new trajectory to explore}
		\State  $u_a(x) \gets \mathbf{0}$ \Comment{Continue exploiting this policy}
		\EndWhile
		\EndFunction
		\Function{EvaluateReward}{}
		\Statex Dynamics is evolving per $x(t)=e^{(A_h+BF )(t-t_{k} )} x(t_{k})$.
		\State $\delta_k(F) \gets \eqref{Equation:deltaF}$; $\delta_k(\mathbf{e}_{pq}) \gets \eqref{Equation:deltae}$
		\State $r_k(Q) \gets \eqref{Equation:rQ}$; $r_k(\mathbf{e}_{pq}) \gets \eqref{Equation:re}$
		\State \textbf{return} $\delta_k(F),\delta_k(\mathbf{e}_{pq}),r_k(Q), r_k(\mathbf{e}_{pq}) $ \Comment{Return $K_i$ free parts}
		\EndFunction
		\Function{Nudge}{}
		\State  $u_a(x) \gets \mathbf{0} + PRNG$ \Comment {Pseudorandom Number Generator}
		\EndFunction
		\Function{Recompute}{$\mathbf{A}$,$\mathbf{b}$} \label{Procedure:UpdateAb}
		
		\State $\forall k, \delta_k(K_i) \gets \eqref{Equation:deltaKi}, \delta_k \gets \eqref{Equation:deltaAll}$
		\State $\forall k, r_k(K_i) \gets \eqref{Equation:rKi},r_k \gets \eqref{Equation:rAll}$
		\If{$ cond(\mathbf{A})<\epsilon  $}
		\State $\mathbf{A},\mathbf{b} \gets$ \Call{ExploitPolicy}{}
		\EndIf
		\State \textbf{return} $\mathbf{A},\mathbf{b}$
		\EndFunction	
	\end{algorithmic}
\end{algorithm}

In \Cref{Problem:TakeOver}, the aim is to find $u_a(x)$ that is optimal after the human operator is removed leaving the autonomy system alone. Thus the underlying ARE is given by
\begin{equation} \label{Equation:ARE-ParallelAutonomyTakeOver} 
0=PA+A^{\intercal}P- PBR^{-1}B^{\intercal}P+Q,
\end{equation}
The optimal takeover policy is given by $u_a(x)=-R^{-1} B^{\intercal}P x$ where $P$ is the stabilizing solution of \eqref{Equation:ARE-ParallelAutonomyTakeOver}. During learning, $u_a = 0$ and $u_0 = u_h$ as shown in \Cref{Algorithm:TakeOverOffPolicy}. This initialization differs from the off-policy implementation for the minimum intervention policy learning. The behavior policy here is $u_h + u_a$, thus ${\Delta(u_h + u_a,u_i)}$ is used in $\delta_k$ in \eqref{Equation:deltaAll} which is another difference from the minimum intervention case.
Lastly, we should point out that the closed-loop dynamics for the first iteration in all three algorithms is the same, thus the off-policy takeover learning can be implemented in parallel to the learning of the minimum intervention policy.

\section{Car-following Example} \label{Section:Example}
We show an application of sLQR to a car-following problem in which a car with a parallel autonomy system is to maintain a particular constant spacing from a leading vehicle, and achieve the same speed.

\begin{figure}[!htb]
	\center{\includegraphics[width=\columnwidth]
		{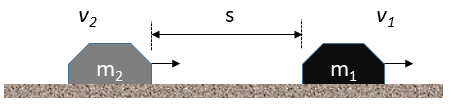}}
	\caption{\label{Figure:CarFollowing} Car Following.}
\end{figure}

The error dynamics are adapted from \cite{LevineAthans1966}
\begin{equation} \label{Equation:CarFollowingErrorDynamics}
	\begin{aligned}
		\dot{x}_1(t) =& -\frac{\alpha_1}{m_1} x_1(t), \\
		\dot{x}_2(t) =& x_1(t) - x_3(t), \\
		\dot{x}_3(t) =& -\frac{\alpha_2}{m_2} x_3(t) + \frac{1}{m_2}  u,
	\end{aligned}
\end{equation}
where $x_1(t) = \tilde{v}_1(t)$, $x_2(t) = \tilde{s}(t)$, $x_3(t) = \tilde{v}_3(t)$ and $u(t) = \tilde{f_2}(t)$. Moreover, $\tilde{v}_1$, $\tilde{v}_2$ are the speed error variables and $\tilde{s}$ is the spacing error variable and $\tilde{f_2}(t)$ is the force error applied to the following car. Let $m_1=m_2=1$ and $\alpha_1=\alpha_2=1$.

We assume that the human operator is applying the following gains $K_h = [0 \; 1 \; -1]$ not known to the autonomy system. Unlike the autonomy system, we assume the human operator has no access to the speed error of the leading vehicle $\tilde{v}_1(t)$, thus 
	\[C_h = \left[  \begin{array}{cl}
	0 & \mathbf{0}\\ \mathbf{0} & \mathbf{I}_{2}
	\end{array}\right].
	\]
Note that the $A_h$ matrix of this dynamical system has repeated eigenvalues and is diagonalizable, thus \Cref{Lemma:DiagonalizableDegenerateNonDegenerate} implies learning cannot happen on the same state-space trajectory. In what follows, we simulate three learning algorithms using the following design parameters $Q = 5\mathbf{I}_{3}$, $M = 1$, $R=10$ and $\tau = 0.01$.

\begin{figure}[!htb]
	\center{\includegraphics[height=6cm,width=\columnwidth,trim={3cm 0 5cm 0},clip]
		{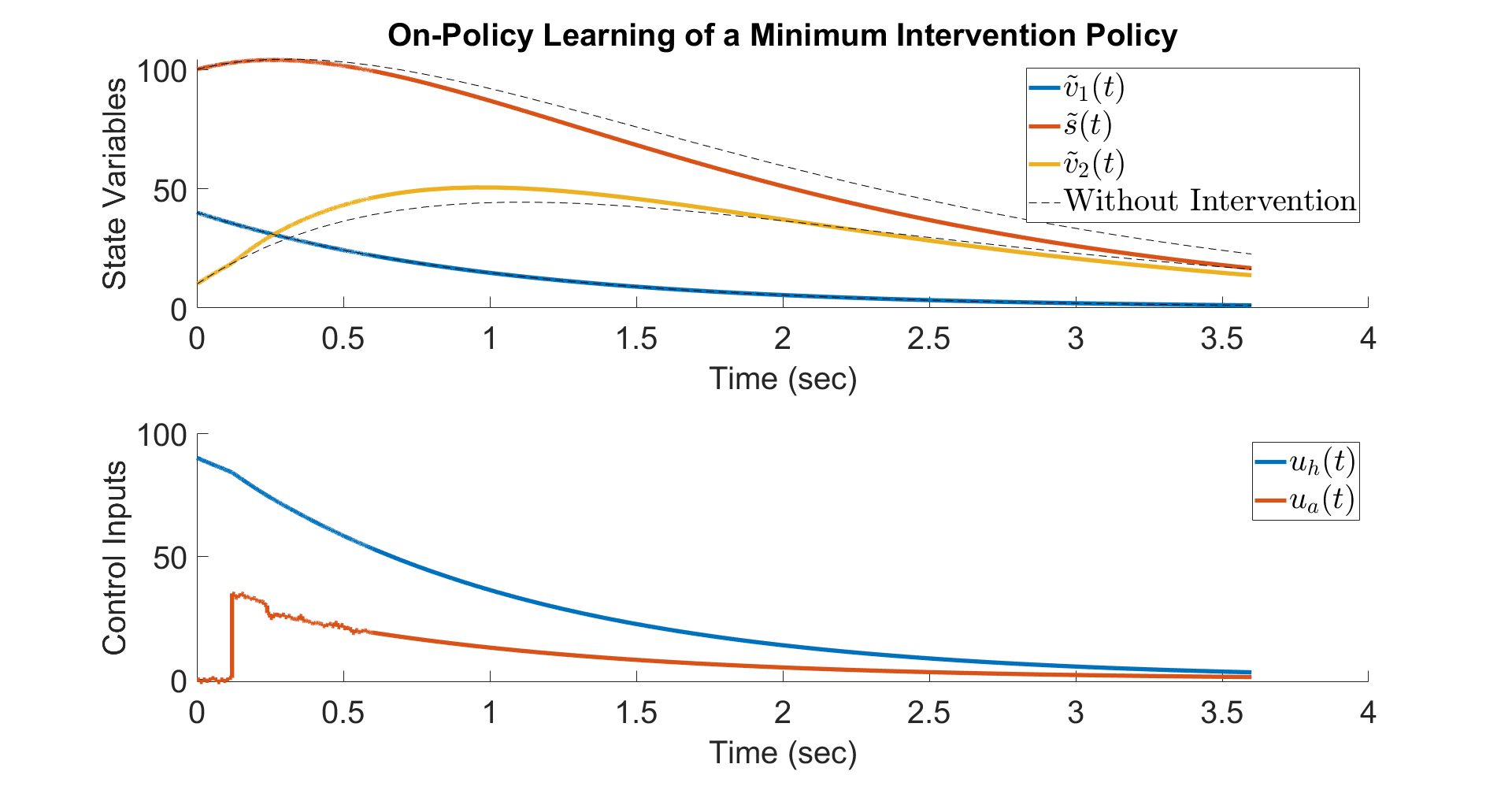}}
	\caption{\label{Figure:OnPolicyIntervention} On-Policy Learning of a Minimum Intervention Policy.}
\end{figure}

\begin{figure}[!htb]
	\center{\includegraphics[height=6cm,width=\columnwidth,trim={3cm 0 5cm 0},clip]
		{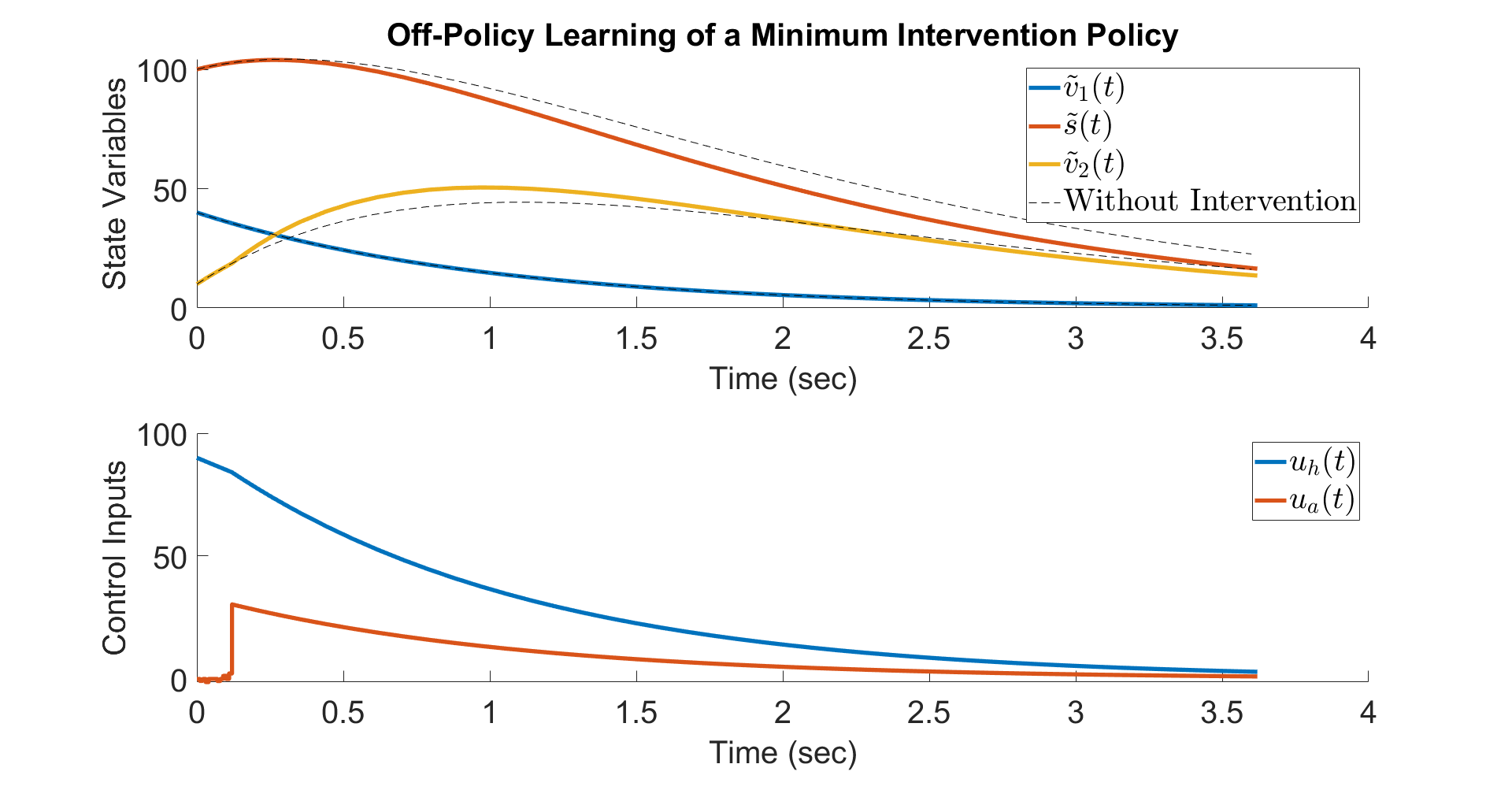}}
	\caption{\label{Figure:OffPolicyIntervention} Off-Policy Learning with Experience Replay of a Minimum Intervention Policy.}
\end{figure}

\begin{figure}[!htb]
	\center{\includegraphics[height=6cm,width=\columnwidth,trim={3cm 0 5cm 0},clip]
		{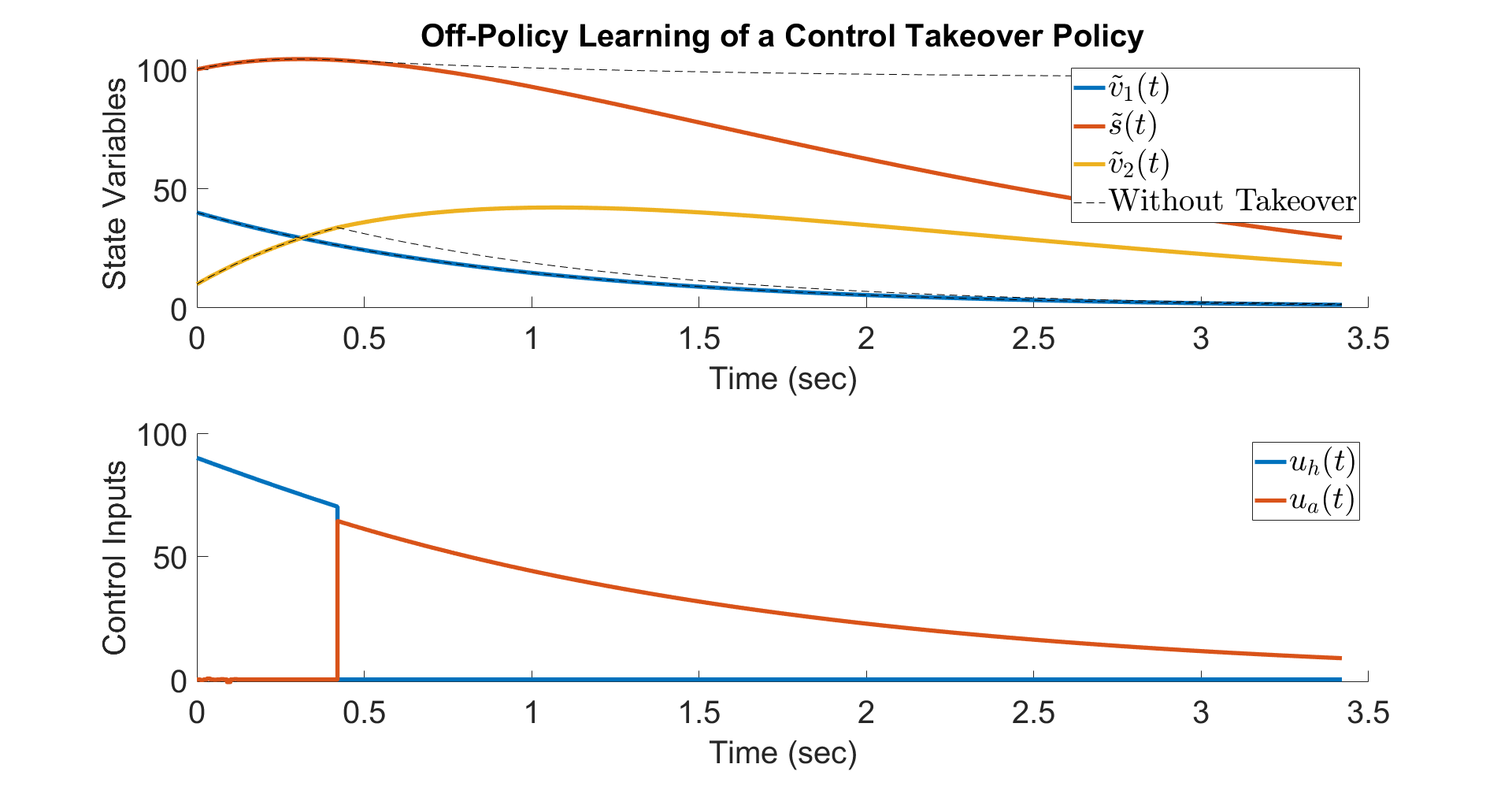}}
	\caption{\label{Figure:OffPolicyTakeover} Off-Policy Learning with Experience Replay of a Control Takeover Policy.}
\end{figure}

As seen in the simulations, the exploration needed is minimal and is enough to change current trajectory to a new one to resume learning. The strength of the nudge is minimal compared to the strength of $u_a$ or $u_h$. Note also that off-policy learning requires less data as the closed-loop remains fixed, and the gains can be recomputed over old trajectory data by experience replay. This causes less interference with the human operator and may be more favorable. MATLAB code to reproduce the results is at \cite{CodeMATLAB}. As can be seen, only the initial states of trajectory segments are selected randomly by a nudge. Once on a trajectory segment, no noise is injected while evaluating the reward.

\section{Conclusion} \label{Section:Conclusion}
The sLQR empowers human operators due to the full access the autonomy system has to the state of the plant. Additionally, the role of exploration ensures minimal interference with the human operator.

%


\end{document}